\RequirePackage{tikz}
\documentclass[%
 letterpaper,
 onecolumn,
  colorlinks,
]{preprint}

\usepackage{cite}

\usepackage[english]{babel}

\usepackage{amsmath,amssymb,amsfonts}
\usepackage{graphicx}
\usetikzlibrary{positioning}
\usepackage{textcomp}
\usepackage{comment,todonotes}

\mathcode`\@="8000
{\catcode`\@=\active \gdef@{\mkern1mu}}

\def\R{{\mathbb{R}}}
\def\C{{\mathbb{C}}}
\def\Ir{{I_r}}

\def\Inr{{I_{n-r}}}
\def\X{{\mathcal{X}}}
\def\CH{{\mathcal{H}}}
\def\tr{{\kern-1pt{\scriptscriptstyle{\top}}\kern-1pt}}
\def\Sys{\mathcal{G}}
\def\Sysred{\Sys_r}

\def\Systrunc{\wt{\Sys}}
\def\trace{{\textrm{t{\kern0.4pt}r}}}
\def\wt#1{\widetilde{#1}}
\def\wh#1{\widehat{#1}}
\def\imunit{{\mathrm{i}}}

\def\hot#1{\textcolor{red}{#1}}

\newtheorem{theorem}{Theorem}[section]
\newtheorem{corollary}{Corollary}[section]
\newtheorem{lemma}{Lemma}[section]
\newtheorem{proposition}{Proposition}[section]
\newtheorem{remark}{Remark}[section]
\newtheorem{definition}{Definition}[section]
\newtheorem{example}{Example}[section]

\newtheorem{proof}{Proof}[section]

\begin{document}

\title{On the balanced truncation error bound and sign parameters from arrowhead realizations}
  
\author[$\dagger$]{Sean Reiter}
\affil[$\dagger$]{
  Department of Mathematics, Virginia Tech,
  Blacksburg, VA 24061, USA.\authorcr
  \email{seanr7@vt.edu}, \orcid{0000-0002-7510-1530}
}

\author[$\ast$]{Tobias Damm}
\affil[$\ast$]{%
  Department of Mathematics, University of Kaiserslautern-Landau (RPTU),
  67653 Kaiserslautern, Germany.\authorcr
  \email{damm@mathematik.uni-kl.de}, \orcid{0000-0002-8550-8174}
}

\author[$\ddagger$]{Mark Embree}
\affil[$\ddagger$]{%
  Department of Mathematics and Division of Computational Modeling and Data
  Analytics, Academy of Data Science, Virginia Tech,
  Blacksburg, VA 24061, USA.\authorcr
  \email{embree@vt.edu}, \orcid{0000-0001-6456-9317}
}

\author[$\ddagger$]{Serkan Gugercin}
\affil[$\ddagger$]{%
  Department of Mathematics and Division of Computational Modeling and Data
  Analytics, Academy of Data Science, Virginia Tech,
  Blacksburg, VA 24061, USA.\authorcr
  \email{gugercin@vt.edu}, \orcid{0000-0003-4564-5999}
}

\shorttitle{On the balanced truncation error bound}
\shortauthor{S. Reiter, T. Damm, S. Gugercin, M. Embree}
\shortdate{\today}
  


\abstract{Balanced truncation and singular perturbation approximation for linear dynamical systems yield reduced-order models that satisfy a well-known error bound involving the Hankel singular values.
We show that this bound holds with equality for single-input, single-output systems, if the sign parameters corresponding to the truncated Hankel singular values are all equal. 
These signs are determined by a generalized state-space symmetry property of the corresponding linear model.
For a special class of systems having arrowhead realizations, the signs can be determined directly from the off-diagonal entries of the corresponding arrowhead matrix.
We describe how such arrowhead systems arise naturally in certain applications of network modeling, and illustrate these results with a power system model that motivated this study.
}
\novelty{}

\maketitle

\section{Introduction}\label{sec:introduction}
Balanced truncation~\cite{mullis1976synthesis,moore1981principal} is a powerful method for reducing linear time-invariant (LTI) dynamical systems, yielding reduced-order models that are guaranteed to be asymptotically stable and satisfy a simple \emph{a priori} $\mathcal{H}_\infty$ error bound involving the Hankel singular values.
The bound is known to hold with equality under certain conditions, such as when the full-order system is state-space symmetric~\cite{LST98}. 

This paper makes two primary contributions. 
First, we show that the balanced truncation error bound holds with equality for a larger family of single-input, single-output (SISO) systems, those for which the truncated part of the model satisfies a particular state-space symmetry in its canonical balanced form. 
This result can also be obtained using results from~\cite{6949058}; see Remark~\ref{remark:timoresult} for details.
The \textit{sign parameters} associated with the Hankel singular values provide the main tool for determining this generalized state-space symmetry of the balanced model. {The signs are a state-space invariant with respect to the underlying system.}
{Secondly, we show that these sign parameters are determined by \emph{any} realization of the full-order model that satisfies this generalized state-space symmetry condition.}
As a consequence of this result, we show how to derive these sign parameters for a special class of systems having \emph{arrowhead realizations}. In this case, the signs can be read off directly from the off-diagonal entries of the corresponding arrowhead matrix.
Systems with arrowhead realizations arise in network models
where the internal state variables interact as in Figure~\ref{figure:arrowheadnetwork}, with the input $u(t)$ and output $y(t)$  directly interacting only with the first state variable, $x_1(t)$. 
We derive sufficient conditions that guarantee a general linear system has such an arrowhead form.

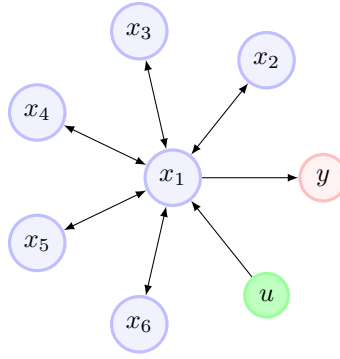
\begin{figure}[b!]
\center
\begin{tikzpicture}[roundnode/.style={circle, draw=blue!25, fill=blue!5, very thick, minimum size=5mm},
outputnode/.style={circle, draw=red!25, fill=red!5, very thick, minimum size=5mm},
inputnode/.style={circle, draw=green!40, fill=green!25, very thick, minimum size=5mm}]
\def \n {6}
\def \radius {2cm}
\def \margin {8}
\node[roundnode] at (360:0) (hub) {$x_1$};
\foreach \s in {2,...,\n}{
    \node[roundnode] at  ({360/(\n+1) * (\s - 1)}:\radius) (\s) {$x_{\s}$};
    \draw[<->, >=latex] (hub)--(\s);
}
\node[inputnode] at  ({360/(\n+1) * (7 - 1)}:\radius) (7) {$u$};
\draw[->, >=latex] (7)--(hub);
\node[outputnode] at  ({360/(\n+1) * (8 - 1)}:\radius) (8) {$y$};
\draw[->, >=latex] (hub)--(8);
\end{tikzpicture}
\caption{\label{figure:arrowheadnetwork} An arrowhead network with $n=6$, with input $u$ and output $y$ restricted to state $x_1$.}
\end{figure}

We first encountered this error bound phenomenon while studying a model from power system dynamics. We observed numerically that the balanced truncation error bound for these models was tight, and  that their canonical balanced realization always exhibited a particular sign symmetry in which the truncated part of the model has a slight generalization of this state-space symmetric form.
We then noted that this sign symmetry could be derived from a natural state-space realization in which the model has an arrowhead form, as shown in Figure~\ref{figure:arrowheadnetwork}.

The rest of the paper proceeds as follows. Section~\ref{sec:balancedtrunc} reviews the necessary details of balanced truncation model reduction, and describes a generalization of state-space symmetry based on the canonical balanced form of a linear system. We also establish the connection between this symmetry and the \textit{sign parameters} corresponding to the system's Hankel singular values.
Section~\ref{sec:result3} contains the first main result, showing that the balanced truncation error bound is tight when the \emph{truncated} system exhibits a slight generalization of state-space symmetry.
This symmetry condition holds when the sign parameters corresponding to the truncated Hankel singular values are equal.
In Section~\ref{sec:gen_sign_symm} we show how to derive a system's sign parameters (and thus the sign symmetry of the associated canonical balanced form) from \emph{any} state-space realization $\Sys$ that satisfies this generalized sign-symmetry condition.
In Section~\ref{sec:specialcase}, we study the arrowhead systems that motivated this work, and apply the result of Section~\ref{sec:gen_sign_symm} to a special class of arrowhead systems. 
Numerical examples follow each result for illustrative purposes, including the power system model that motivated our investigation. Section~\ref{sec:conc} concludes the paper.

\section{Balanced truncation model reduction}
\label{sec:balancedtrunc}

Consider the SISO LTI dynamical system
\begin{align}\label{eq:system}
\Sys:
\left\{
\begin{array}{rcl}x'(t)  & = &  A@x(t) + 
b@u(t) \\
y(t) &= & c@x(t) +du(t),
\end{array}
\right.
\end{align}
having the  transfer function
$$G(s) = c@(sI-A)^{-1}b+d,$$
where $A\in\mathbb{R}^{n\times n}$, $b\in\mathbb{R}^{n\times 1}$, $c\in\mathbb{R}^{1\times n}$, and $d\in\mathbb{R}$.\ \ Throughout we assume that
$\mathcal{G}$ is asymptotically stable, i.e., the eigenvalues of $A$ have negative real parts, and that $\mathcal{G}$ is minimal, i.e., it is reachable and observable. 
The $\mathcal{H}_\infty$-norm of {the LTI system} $\mathcal{G}$ is defined as {the $\CH_\infty$ norm of its transfer function $G(s)$~\cite[Sec.~5.1.3]{Ant05}, i.e.}
\begin{align}
\label{eq:Hinfnorm} {\|G\|_{\CH_\infty}}= \sup_{\omega\in\mathbb{R}}  | G(\imunit@\omega) |,
\end{align}
where {$\imunit\in\C$ denotes the imaginary unit}.
The Lyapunov equations
\begin{align} \label{eq:lyap}
    A\mathcal{P}+\mathcal{P}A^{\tr} + b@b^{\tr} =0\quad\mbox{and}\quad
    A^{\tr}\mathcal{Q}+\mathcal{Q}A + c^{\tr}  c=0
\end{align}
have unique positive definite solutions $\mathcal{P}$ and $\mathcal{Q}$
called the reachability and observability Gramians of $\Sys$.\ \ 
We say the system $\Sys$ is \textit{principal-axis balanced} if $\mathcal{P}=\mathcal{Q}=\Sigma$, where 
\begin{align} \label{eq:hsvs}
\Sigma=\textrm{diag}(\sigma_1 I_{m_1}, \dots,\sigma_q I_{m_q} ),
\end{align}
with $\sigma_1>\sigma_2>\dots>\sigma_q>0$ and $m_1+\cdots+m_q=n$.
Here $\sigma_1,\ldots,\sigma_q$ are the \textit{Hankel singular values} of $\mathcal{G}$. 
For brevity, throughout we use the term \textit{balanced} as shorthand for \emph{principal-axis balanced}.
The state-space representation in balanced coordinates is called a \textit{balanced realization} of $\Sys$.\ \ In these balanced coordinates, states that are difficult to reach are also difficult to observe; such states are associated with the smallest Hankel singular values. Balanced truncation reduces the model by removing these components of the state space.
This reduction can be expressed compactly by writing the system in block form, as given in the following theorem.
The asymptotic stability result is due to
Pernebo and Silverman~\cite{PerS82} and the error bound~\eqref{eq:bound} is due to Enns~\cite{enns1984model}.
\begin{theorem} 
Let $\Sys$ be an order-$n$ minimal and asymptotically stable dynamical system having the balanced realization
\begin{align} \label{eq:partitionedbal}
\left[\begin{array}{c|c}
A & b \\\hline
c & d
\end{array}
\right]=\left[
\begin{array}{cc|c}
A_{11}& A_{12}  & b_1\\
A_{21} & A_{22} & b_2\\
\hline
c_1 & c_2 &d
\end{array}
\right],
 \end{align}
where the state space matrices are partitioned according to the system Gramian $\Sigma = \textrm{diag}(\Sigma_1,\Sigma_2)$, with
\begin{align}
    \label{eq:hsv1}\Sigma_1 &= \textrm{diag}(\sigma_1 I_{m_1},\dots,\sigma_k I_{m_k}), \\
    \label{eq:hsv2a}\Sigma_2 &= \textrm{diag}(\sigma_{k+1} I_{m_{k+1}},\dots,\sigma_q I_{m_q}),
\end{align}
where $r:=m_1+\cdots+m_k$. Then the $r$th order reduced model obtained via balanced truncation,
\begin{align} \label{eq:redsys}
\Sysred:
\left\{
\begin{array}{rcl}x_r'(t)&=& A_{11}@x_r(t) + 
b_1@u(t)\\
y_r(t) &=& c_1@x_r(t) +d@u(t),
\end{array}
\right.
\end{align}
having the transfer function
$$G_r(s) = c_1(s\Ir-A_{11})^{-1}b_1 +d,$$
is asymptotically stable and satisfies the error bound
\begin{align} \label{eq:bound}
\|{G}-{G_r}\|_{\mathcal{H}_\infty}
\leq 2@ (\sigma_{k+1}+\dots +\sigma_q).
\end{align}
\end{theorem}
Equality is known to hold in~\eqref{eq:bound} when only one Hankel singular value is truncated, that is, $\Sigma_2 = \sigma_qI_{m_q}$.
We next explore other situations in which the bound~\eqref{eq:bound} holds
with equality, thus giving an exact formula for the error 
in the reduced order model.

\subsection{The sign symmetry of balanced realizations}

We next describe an important sign symmetry property of balanced systems characterized in~\cite{wilson1983symmetry,ober1987balanced}, which will be essential for our main result in Section~\ref{sec:result3}.

\begin{theorem} 
\label{thm:kumar} Let $\Sys$ be an asymptotically stable and  minimal SISO LTI system. Then $\Sys$ has a balanced realization satisfying the sign-symmetric form
\begin{align} \label{eq:signA}
    A = SA^{\tr}S \quad \mbox{and} \quad
    b = (cS)^{\tr}, 
\end{align}
where $S=\textrm{diag}(s_1,s_2,\dots,s_n)$ and 
{$s_i\in\{\pm 1\}$} for $i=1,\ldots,n$.
\end{theorem}
We refer to {$s_1,\ldots,s_n$} as the \textit{sign parameters} of $\Sys$.\ \   
Ober~\cite{ober1987balanced,ober1989balanced} shows that every asymptotically stable, minimal, SISO LTI system is equivalent to a balanced system with a realization satisfying~\eqref{eq:signA}, which we call the \textit{canonical form} of a balanced system.
{In the case of distinct Hankel singular values, i.e. $q=n$, we can explicitly} apply the formula~\cite[eq.~(7.24)]{Ant05} to construct the canonical form of $\Sys$ having the state-space matrices $A,b$, and $c$ as 
\begin{align}
    \label{eq:canonicalform}a_{ij}=\frac{-\gamma_i\gamma_j}{s_is_j\sigma_i+\sigma_j},
    \quad b_i = \gamma_i,
    \quad c_i = s_i\gamma_i,\qquad i,j=1,\ldots,n,
\end{align}
where $a_{ij}$ denotes the $(i,j)$th entry of $A$.\ \  The system is completely determined by the distinct Hankel singular values $\sigma_1,\ldots,\sigma_n>0$, the signs $s_1,\ldots,s_n \in \{\pm 1\}$, and the entries $\gamma_1,\ldots,\gamma_n$ of $b\in\R^n$.
A similar formula holds for systems with repeated Hankel singular values; see~\cite[eq.~(7.26)]{Ant05}.
We include this formula to emphasize that the sign symmetry~\eqref{eq:signA} is not merely an artifact of the canonical balanced realization; these signs arise explicitly in the parameterization of all balanced linear systems~\cite{ober1987balanced}.
Analytically, each sign parameter is naturally associated with one of the system's Hankel singular values; that is, each $s_i$ corresponds to $\sigma_j$ for some $j=1,\ldots,q$. 
(Consequently, multiple sign parameters can correspond to the same Hankel singular value, as is the case with repeated Hankel singular values; see~\cite[Section~2.7]{maciejowski1988balanced} for specific details.) We additionally assume that the sign parameters in $S$ are ordered such that the associated Hankel singular values are non-increasing.
As we will show in Section~\ref{sec:result3}, the importance of these sign parameters (and, equivalently, the sign-symmetry of the canonical balanced form~\eqref{eq:signA}) is that they provide sufficient conditions for determining whether the $\CH_\infty$ error bound~\eqref{eq:bound} holds with equality.
In Section~\ref{sec:gen_sign_symm}, we will show how to determine these sign parameters from \emph{any} (not necessarily balanced) realization of $\Sys$ satisfying the sign-symmetry structure~\eqref{eq:signA}.

\begin{remark}
\label{remark:bal_ordering}
Any balanced realization of a system $\Sys$ is unique up to orthogonal transformations of the state space~\cite{moore1981principal}, permitting multiple balanced realizations of $\Sys$ that obey the same sign symmetry as~\eqref{eq:signA} but with different permutations of the signs on the diagonal of $S$.\ \ 
In these realizations, the associated balanced coordinates are generally \textit{not} ordered in decreasing significance, in contrast to the canonical form.
\end{remark}

Now suppose $\Sys$ is balanced, having the canonical form \eqref{eq:signA}.
Define $r:= m_1 + \cdots  + m_k$ for $1\le k < q$, and partition the sign matrix as 
\[ S=\textrm{diag}(S_1,S_2),\]
where 
\begin{align}
    S_1 = \textrm{diag}(s_1,\dots,s_k) \quad \mbox{and} \quad
     S_2 = \textrm{diag}(s_{k+1},\dots,s_n).
\end{align}
Partition $A$, $b$, and $c$ {as in~\eqref{eq:partitionedbal}}, giving the sign symmetries
\begin{align}
\begin{split}
     \label{eq:part_symm1}&A_{11}=S_1A_{11}^{\tr}S_1,\ \ \ A_{12}=S_1A_{21}^{\tr}S_2,\ \ \ A_{21}=S_2A_{12}^{\tr}S_1,\\
    &A_{22}=S_2A_{22}^{\tr}S_2,\ \ \ b_1 = (c_1S_1)^{\tr},\ \mbox{and} \ b_2 = (c_2S_2)^{\tr},
\end{split}
\end{align}
which follow from direct multiplication in \eqref{eq:signA}. 
{Since the reduced model obtained via balanced truncation is independent of the initial system realization, for the {time being} we will assume, without loss of generality, that
$\Sys$ is already balanced with the realization given in~\eqref{eq:signA}.}
Therefore, the Lyapunov equations
in \eqref{eq:lyap} both have the solution $\mathcal{P}=\mathcal{Q}=\Sigma$:
\begin{align}
    \label{eq:lyap1} A@\Sigma + \Sigma A^{\tr} +b@b^{\tr}=0\quad\mbox{and}\quad
 A^{\tr}\Sigma + \Sigma A + c^{\tr}c&=0.
\end{align}


\begin{definition} \label{def:symm}
If {$\Sys$ has a (not necessarily balanced) realization satisfying~\eqref{eq:signA} with $S=I$}, then $\Sys$ is called \emph{state-space symmetric}.  In this case, $A=A^{\tr}$ and $b=c^{\tr}$. 
\end{definition}

As stated earlier, the $\mathcal{H}_\infty$ error bound~\eqref{eq:bound} for balanced truncation holds with equality if only one singular value is truncated~\cite{enns1984model}. However, state-space symmetric systems have the property that 
the error bound~\eqref{eq:bound} holds with equality \emph{for any truncation order}~\cite{LST98}.
We extend this result to a more general setting next.

\section{Error bound results}  \label{sec:result3}

In this section, we show that the $\mathcal{H}_\infty$ error bound for balanced truncation holds with equality for a more general class of systems than the state-space symmetric ones just described. In Section~\ref{sec:specialcase} we will provide an example from this new class of systems that arises naturally in power system  modeling.
In accordance with the partitioned balanced realization
 \eqref{eq:partitionedbal}, 
we define the \textit{truncated system} 
\begin{align}
    \label{eq:truncsys} \Systrunc:
\left\{
\begin{array}{rcl}\wt{x}'(t)
   &=&A_{22}@\wt{x}(t) + b_2@u(t)\\
\wt{y}(t)  &=&  c_2@\wt{x}(t),
\end{array}
\right.
\end{align}
having transfer function $\wt{G}(s)=c_2(s{\Inr}-A_{22})^{-1}b_2$ and {system Gramian}~\eqref{eq:hsv2a}. We call~\eqref{eq:truncsys} the {``truncated system''} because its state $\wt{x}(t)\in\mathbb{R}^{n-r}$ corresponds to the truncated states in the balanced realization~\eqref{eq:signA}. Note from~\eqref{eq:part_symm1} that $\Systrunc$ satisfies the sign symmetries {$A_{22} = S_2A^{\tr}_{22}S_2$ and $b_2=(c_2 S_2)^{\tr}$}.  
For the results below, we will allow the signs in $S_1$ {of a system $\Sys$} to vary, but assume that either $S_2=I_{n-r}$ or $S_2=-I_{n-r}$. In other words, we do not assume that $\Sys$ is state-space symmetric; we only assume that the sign parameters corresponding to the \emph{truncated} Hankel singular values are the same. 
In such cases, the truncated system $\wt \Sys$ obeys the state-space symmetry $A_{22}=A_{22}^{\tr}$ and $b_2 = \pm c_2^{\tr}$.

\subsection{Generalized conditions for the balanced truncation error bound to hold with equality}
\label{ss:btresult}
\begin{theorem} \label{thm:main}
Let $\Sys$ be an order-$n$ asymptotically stable, minimal, balanced SISO system as  in \eqref{eq:system}, with its matrix of Hankel singular values 
$\Sigma = \textrm{diag}(\Sigma_1,\Sigma_2)$ 
partitioned as in \eqref{eq:hsv1} and \eqref{eq:hsv2a};
set $r:= m_1+\cdots+m_k$. 
Let $\Sysred$ denote the order-$r$ approximation of
$\Sys$ via balanced truncation, as in \eqref{eq:redsys}.
Conformally partition the sign matrix, 
$S = \textrm{diag}(S_1,S_2)$, with $S_1 \in \R^{r\times r}$ and $S_2 \in \R^{(n-r)\times(n-r)}$.
If all the signs in $S_2$ are the same, i.e.,
\begin{equation}
    \label{eq:s2}S_2 = \textrm{diag}\,(+1, \dots,+1)\,\,\, \textrm{or} \,\,\,S_2 = \textrm{diag}\,(-1, \dots,-1),
\end{equation}
then the truncated Hankel singular values are distinct,
\begin{equation}
\label{eq:hsv2}
  \Sigma_2 = \textrm{diag}(\sigma_{k+1},\dots,\sigma_q),
\end{equation}
and  $\Sysred$ in~\eqref{eq:redsys} 
achieves the {balanced truncation} error bound~\eqref{eq:bound}, {i.e.}
$$\left\|{G}-{G_r}\right\|_{\mathcal{H}_\infty}=2(\sigma_{k+1}+\cdots+\sigma_q).$$
\end{theorem} 

\begin{proof}
Because $\Sysred$ is obtained via balanced truncation, the error system satisfies the upper bound~\eqref{eq:bound}, and so it suffices to show that this bound is attained for the particular frequency $\omega=0$, 
i.e., 
\begin{align*}
  2 @@ \trace (\Sigma_2) &= |G(0)-G_r(0)|=\big|c@A^{-1}b-c_1A_{11}^{-1}b_1\big|.
\end{align*}
First, note that the difference $G_r(0)-G(0) = c@A^{-1}b-c_1A_{11}^{-1}b_1$ in the right-hand side of the above error expression can be expressed in terms of the Schur complement
\begin{align*}
A/A_{11}:=A_{22}-A_{21}A_{11}^{-1}A_{12}.
\end{align*}
To see this, define the  matrices
$$L=\left[
  \begin{array}{cc}
    I_r&0\\-A_{21}A_{11}^{-1}&I_{n-r}
  \end{array}
\right] \ \ \ \textrm{and} \ \ \ R=\left[
  \begin{array}{cc}
    I_r&-A_{11}^{-1}A_{12}\\0&I_{n-r}
  \end{array}
\right].$$
One can verify via the formula for the inverse of a $2\times 2$ block matrix~\cite[Theorem~2.1]{lu2002inverses} that $A^{-1} = R\,\textrm{diag}\,(A_{11}^{-1},(A/A_{11})^{-1})L$, and so
\begin{align}\nonumber
  G_r(0)-G(0)&=
c\left(A^{-1}-\left[
  \begin{array}{cc}
    A_{11}^{-1}&0\\0&0
  \end{array}
\right]\right)b\\
\nonumber
&=
c\left(R\left[
  \begin{array}{cc}
    A_{11}^{-1}&0\\0&(A/A_{11})^{-1}
  \end{array}
\right]L-R\left[
  \begin{array}{cc}
    A_{11}^{-1}&0\\0&0
  \end{array}
\right]L\right)b\\
\nonumber
&=
 cR \left[
  \begin{array}{cc}
 0&0\\ 0&(A/A_{11})^{-1}
  \end{array}
  \right]
    Lb\\&=\trace \left(\left[ \begin{array}{cc}
 0&0\\ 0&(A/A_{11})^{-1}
  \end{array}
  \right]L@b@c@R\right).\label{eq:LbcR}
\end{align}
This last equality follows from the fact that {$G(0)-G_r(0)$} is a scalar, and that the trace of the product of matrices is invariant under cyclic permutation.
By the sign symmetry~\eqref{eq:signA}, we can transform the Lyapunov equation
$-b@b^\tr=A\Sigma+\Sigma A^\tr $ to 
\begin{align}
\label{eq:modified_lyap}
  -b@c=-b@b^\tr S&=A\Sigma S+\Sigma A^\tr S =A(\Sigma S) +\Sigma (SAS) S =A(\Sigma S)+(\Sigma S)A.
\end{align}
Using the partitioning of $A$ as in~\eqref{eq:partitionedbal}, block matrix multiplication reveals $LAR=\textrm{diag}\,(A_{11},\,A/A_{11}).$  
Multiplying~\eqref{eq:modified_lyap} on the left and right by $L$ and $R$ respectively, we can exploit the triangular structure of these matrices to obtain 
\begin{align*}
  -L@b@c@R&= L{\left(A(RR^{-1})(\Sigma S)+(\Sigma S)(L^{-1}L)A\right)}R\\
  &=(LAR)R^{-1}(\Sigma S)R+L(\Sigma S)L^{-1}(LAR)\\
&        =\left[
          \begin{array}{cc}
            A_{11}&0\\0&A/A_{11}
          \end{array}
\right]\left[
          \begin{array}{cc}
            S_1\Sigma_1&\star\\ 0 &S_2\Sigma_2
          \end{array}
\right]+ \left[\begin{array}{cc}
            S_1\Sigma_1&0\\\star &S_2\Sigma_2
          \end{array}
\right]\left[
          \begin{array}{cc}
            A_{11}&0\\0&A/A_{11}
          \end{array}
\right].
\end{align*}
(Entries indicated by $\star$ are irrelevant.) Inserting this expression into \eqref{eq:LbcR} {and again using invariance of the trace under cyclic permutations, we obtain}
\begin{align*}
    G_r(0)-G(0) &= 
\trace\left(\left[\begin{array}{cc} 0 & 0 \\ 0 &(A/A_{11})^{-1} \end{array}\right]\left[\begin{array}{cc} A_{11} & 0 \\ 0  & A/A_{11} \end{array}\right]\left[\begin{array}{cc} S_1\Sigma_1 & \star \\ 0 & S_2 \Sigma_2 \end{array}\right]\right)
\\
&\phantom{=} {} + \trace\left(\left[\begin{array}{cc} 0 & 0 \\ 0 &(A/A_{11})^{-1} \end{array}\right]\left[\begin{array}{cc} S_1\Sigma_1 & 0 \\ \star & S_2 \Sigma_2 \end{array}\right]\left[\begin{array}{cc} A_{11} & 0 \\ 0  & A/A_{11} \end{array}\right]\right)
\\[5pt]
&= 2@@ \trace\left(\left[\begin{array}{cc} 0 & 0 \\ 0 &S_2 \Sigma_2 \end{array}\right]\right)
\\[5pt]
          &=2@@\trace (S_2\Sigma_2).
\end{align*}
Thus, if $S_2 = \pm I_{n-r}$ then $|G(0)-G_r(0)| = 2@@\trace(\Sigma_2)$ and hence the bound~\eqref{eq:bound} is sharp.
To see that the Hankel singular values are distinct, note that the  conclusion that $|G(0)-G_r(0)|=2@@\trace(\Sigma_2)$, along with the fact that $\|{G}-{G_r}\|_{\CH_\infty}\leq 2\left(\sigma_{k+1}+\cdots+\sigma_q\right)$, necessarily implies that $m_{k+1}=\cdots = m_q = 1$. Otherwise, the magnitude $2@@\trace(\Sigma_2)$ would exceed the bound in~\eqref{eq:bound}. We thus conclude that
$$\|{G}-{G_r}\|_{\CH_\infty}=2\left(\sigma_{k+1}+\cdots+\sigma_{q}\right),$$
completing the proof.
\end{proof}

This result neatly generalizes the sharpness of the balanced truncation error bound when only one Hankel singular value is truncated: in this case the hypothesis of Theorem~\ref{thm:main} is always satisfied, since the truncated system contains only one distinct sign parameter.

\begin{remark}
\label{remark:timoresult}
Theorem~\ref{thm:main} can also be deduced by adapting results from~\cite{6949058}, which proves an $\mathcal{H}_\infty$ lower bound on the error in balanced truncation model reduction for systems with semi-definite Hankel operators. 
While~\cite[Proposition~13]{6949058} is stated for systems that are semi-definite, its proof only requires that the \emph{truncated system} be semi-definite. With this insight, we can apply~\cite[Corollary~14]{6949058} to the systems in Theorem~\ref{thm:main}, showing that the lower bound on the balanced truncation error in~\cite[Corollary~14]{6949058} holds with equality. 
\end{remark}

\subsection{An example with flipped signs and a strict bound}
We illustrate Theorem~\ref{thm:main} with a synthetic example showing how the balanced truncation error bound holds with equality when the truncated system obeys the sign consistency in~\eqref{eq:s2}. 
(Section~\ref{sec:specialcase} includes a power system example that further illustrates this result, along with those in Section~\ref{sec:gen_sign_symm}.)
\begin{example}
\label{ex:thm1}
We construct the system $\Sys$ of order {$n=4$} in its canonical form~\eqref{eq:canonicalform}.
Start by specifying the Hankel singular values 
\[\Sigma = \textrm{diag}(10^1,10^0,10^{-1},10^{-2}),\]
the corresponding sign parameters
\[S=\textrm{diag}(1,1,-1,-1),\]
{and the entries $\gamma_i$ of $b\in\R^4$}
\[{\gamma_1=1,\ \gamma_2=2,\ \gamma_3=3, \ \gamma_4 = 4.}\]
Since the Hankel singular values of $\Sys$ are distinct, we can apply
the formula~\eqref{eq:canonicalform}
to construct the realization explicitly.
The {resulting} system $\Sys$ is  asymptotically stable, minimal, and balanced, having the canonical form
\begin{align*}
A&= \left[\begin{array}{rr|rr}
  -0.05&   -0.18 &   0.30 &    0.40\\
    -0.18&  {-2.00}&  {6.67} &  {8.08}\\\hline
   -0.30&   {-6.67} &  \hot{-45.00} &  \hot{-109.09}\\
   -0.40 &  {-8.08} &  \hot{-109.09} & \hot{-800.00}
\end{array}\right],\\
b &= \left[
\begin{array}{r}   1\\
    {2}\\\hline
    \hot{3}\\
    \hot{4 }\end{array}\right],\ \ \
   c^{\tr} =  \left[\begin{array}{r}  1\\
   {2}\\\hline
   \hot{ -3}\\
   \hot{- 4} \end{array}\right],\ \ \ d=0.
\end{align*}
We compute reduced order models via balanced truncation of orders $r=1,2,3$. The partition of $\Sys$ above highlights the truncation order {$r=2$} to expose the sign symmetry of the truncated system.
Table~\ref{tab:system_errors} compares the $\mathcal{H}_\infty$-norm of the  error system to the balanced truncation upper bound~\eqref{eq:bound}. For reduction to orders $r=2$ and $r=3$, the condition~\eqref{eq:s2} is met, 
and the balanced truncation bound holds with equality, as guaranteed by Theorem~\ref{thm:main}. However for reduction to order $r=1$, the truncated system does not obey the required sign consistency~\eqref{eq:s2}, and the upper bound \eqref{eq:bound} holds with a \emph{strict inequality}.
\begin{table}[hh]
    \caption{\label{tab:system_errors}
     $\mathcal{H}_\infty$ norm of the error system, compared to the balanced truncation upper bound~\eqref{eq:bound} for a  system where Theorem~\ref{thm:main} holds for $r=2$ and $3$, but not $r=1$.}
    \centering
    \vspace{4mm}
    \begin{tabular}{c|cc}
    & $\|{G}-{G_r}\|_{\mathcal{H}_\infty}$ & $2(\sigma_{r+1}+\dots+\sigma_n)$\\[2pt] \hline 
   \rule[-4pt]{0pt}{12pt}
           $r=1$ & $1.780\times10^0\phantom{-}$ & $2.220\times10^0\phantom{-}$ \\
     \rule[-4pt]{0pt}{12pt}
           $r=2$ & $2.200\times10^{-1}$ & $2.200\times10^{-1}$
           \\
      \rule[-4pt]{0pt}{12pt} 
           $r=3$ & $2.000\times 10^{-2}$ & $2.000\times 10^{-2}$\\
        \hline
    \end{tabular}
\end{table}
\end{example}

\subsection{Extension to singular perturbation balancing}

As opposed to truncating the state $\wt{x}(t)$ corresponding to $\Sigma_2$ in the balanced form \eqref{eq:partitionedbal}, one can reduce the model via the \emph{singular perturbation approximation} by setting $\wt{x}'(t) = 0$~\cite{liu1989singular}. Starting with the balanced realization~\eqref{eq:partitionedbal}, 
the order-$r$ \textit{singular perturbation balancing approximation} of $\Sys$ is
\begin{align}
\label{eq:SPA}
\Sysred^{@\textrm{sp}}:\left\{
\begin{array}{rcl}x_r'(t) &=& A^{\textrm{sp}}_r@x_r(t) + 
b^{\textrm{sp}}_r@u(t)\\[3pt]
y_r(t) &=& c^{\textrm{sp}}_r@x_r(t) +d^{@\textrm{sp}}_r@u(t),
\end{array}
\right.
\end{align}
{having the transfer function}
\begin{align*}
    {G_r^{@\textrm{sp}}(s)=c_r^{\textrm{sp}}(sI_r-A_r^{\textrm{sp}})^{-1}b_r^{\textrm{sp}} + d_r^{@\textrm{sp}},}
\end{align*}
where
\begin{align}
\begin{split}
    A^{\textrm{sp}}_r &= A_{11}-A_{21}A_{22}^{-1}A_{12},\ \ \  
    b^{\textrm{sp}}_r = b_1-A_{12}A_{22}^{-1}b_2\\
    c^{\textrm{sp}}_r &= c_1-c_2A_{22}^{-1}A_{21}, \kern22.5pt d^{@\textrm{sp}}_r = d- c_2A_{22}^{-1}b_2.
\end{split}
\end{align}
The reduced model $\Sysred^{@\textrm{sp}}$ also satisfies the $\mathcal{H}_\infty$ error bound \eqref{eq:bound}; see~\cite[Thm.~3.2]{liu1989singular}.

As a consequence of Theorem 3.1, we will show that the $\mathcal{H}_\infty$ error bound~\eqref{eq:bound} also holds with equality when performing singular perturbation balancing, provided the sign parameters of the system $\Sys$ satisfy~\eqref{eq:s2}. {Opmeer and Reis make this observation in~\cite[Section~V.A]{6949058}, also using the concept of reciprocal system as in the proof below, but in a slightly different manner.}
\begin{theorem}
Let $\Sys$ be an order-$n$ asymptotically stable, minimal, balanced SISO system as in~\eqref{eq:system} with its matrices of Hankel singular values $\Sigma=\textrm{diag}(\Sigma_1,\Sigma_2)$ and sign parameters $S=\textrm{diag}(S_1,S_2)$ partitioned as in Theorem~\ref{thm:main}. Let $\Sysred^{@\textrm{sp}}$ be the singular perturbation approximation of $\Sys$ given by~\eqref{eq:SPA}, truncated after the $k$th distinct Hankel singular value and having order $r:=m_1+\cdots+m_k$.
If all the signs in $S_2$ are the same, as in~\eqref{eq:s2}, then $\Sysred^{@\textrm{sp}}$ in~\eqref{eq:SPA} achieves the error bound~\eqref{eq:bound}:
$$\big\|{G}-{G_r^{\textrm{sp}}}\big\|_{\mathcal{H}_\infty}=2(\sigma_{k+1}+\cdots+ \sigma_q).$$
\end{theorem}
\begin{proof} Without loss of generality assume that $\Sys$ is balanced, having the canonical form~\eqref{eq:signA}. It is shown in~\cite[Thm.~3.2]{liu1989singular} that the model reduction error from the $r$th order singular perturbation approximation to $\Sys$ can be written as
$$\big\|{G}-{G_r^{\textrm{sp}}}\big\|_{\mathcal{H}_\infty}
=\big\|{\wh{G}}-{\wh{G}_r}\big\|_{\mathcal{H}_\infty},$$
where {$\wh{G}$ is the transfer function of }the \textit{reciprocal system} of $\Sys$, {denoted by $\wh{\Sys}$ and} given by the realization 
$$\wh{A}=A^{-1},\ \ \ \wh{b}=A^{-1}b,\ \ \ \wh{c}=cA^{-1},$$
and {$\wh{G}_r$ is the transfer function of} $\wh{\Sys}_r$, the $r$th order balanced truncation reduced model for $\wh{\Sys}$.\ \  Then~\cite[Lemma~3.1]{liu1989singular} states that the given realization of $\wh{\Sys}$ is balanced and has Gramian $\Sigma$, and so the Hankel singular values of $\wh{\Sys}$ are the same as those of the original system~$\Sys$.\ \ 
The reciprocal system obeys the same sign symmetry as the original system, that is, $\wh{A}=S\wh{A}^{@@\tr}S$ and $\wh{b}=(\wh{c}S)^{\tr}$, where $S$ is the sign matrix of $\Sys$: inverting both sides of $A=SA^{\tr}S$ shows that $\wh{A}=S\wh{A}^{@@\tr}S$; additionally we see
$$\wh{b} = A^{-1}b = A^{-1}(cS)^{\tr} = A^{-1}Sc^{\tr}=SA^{-\tr}c^{\tr}=S@\wh{c}^{@\tr}.$$
It follows that the submatrices of the reciprocal system partitioned according to \eqref{eq:partitionedbal} satisfy the same sign symmetries as in \eqref{eq:part_symm1}. Thus applying the result of Theorem 3.1 to $\wh{\Sys}$ and $\wh{\Sys}_r$, we conclude 
$$\big\|{G}-{G_r^{@\textrm{sp}}}\big\|_{\mathcal{H}_\infty}
=\big\|{\wh{G}}-{\wh{G}_r}\big\|_{\mathcal{H}_\infty}=2(\sigma_{k+1}+\dots+\sigma_q),$$
completing the proof.
\end{proof}

\section{On the sign parameters of the Hankel singular values}

\label{sec:gen_sign_symm}

Theorem~\ref{thm:main} shows that the balanced truncation $\mathcal{H}_\infty$ error bound holds with equality when the sign parameters corresponding to the truncated Hankel singular values are identical.
In this section, we show that these sign parameters (and thus the associated sign symmetry of the canonical balanced form) are determined up to a permutation by \emph{any} (not necessarily balanced) realization of a SISO system $\Sys$ that satisfies the generalized sign-symmetry condition~\eqref{eq:signA}.
In Section~\ref{sec:specialcase}, we strengthen this result for a select class of systems with \emph{arrowhead realizations}, which includes an example from power systems that motivated this study.

\subsection{Systems with generalized state-space symmetry}
Theorem~\ref{thm:kumar} dictates that every asymptotically stable and minimal linear system $\Sys$ has a \emph{balanced realization} that satisfies the sign-symmetry condition~\eqref{eq:signA}.
We next consider systems having \emph{any} realization that satisfies this sign-symmetry property. That is, we consider SISO LTI systems having a realization such that
\begin{align}
    \label{eq:gen_symm}
    A = \wh{S}A^{\tr}\wh{S} \quad \mbox{and} \quad
    b = (c\wh{S})^{\tr}, 
\end{align}
where $\wh{S}=\mbox{diag}(\wh{s}_1,\wh{s}_2,\ldots,\wh{s}_n)$ and $\wh{s}_i\in\{\pm 1\}$ for $i=1,\ldots,n$.
(This sign-symmetry condition can be viewed as a generalization of the state-space symmetric form in Definition~\ref{def:symm}.)
Note that we do not require the realization~\eqref{eq:gen_symm} to be balanced, and thus the signs that populate $\wh{S}$ are \emph{not} necessarily the sign parameters $\{s_i\}$ of $\Sys$ defined in~\eqref{eq:signA}.

Towards proving the main result of this section, we show how to access the sign parameters of a general linear system via its \emph{cross Gramian}.
The {cross Gramian} of an asymptotically stable SISO system $\Sys$ is the unique solution $\mathcal{X}$ to the Sylvester equation
\begin{align}
    \label{eq:sylveq}
    A\mathcal{X}+\mathcal{X}A+b@c=0.
\end{align}
{Since SISO systems are trivially symmetric, the system Gramians satisfy the relationship $\X^2=\mathcal{P}\mathcal{Q}$~\cite{fernando1985cross}, from which it follows that the eigenvalues of $\X$ are real.}
The cross Gramian enjoys a helpful property that will enable us to determine a system's sign parameters.
For an order-$n$ asymptotically stable and minimal system $\Sys$, denote the eigenvalues of the cross Gramian $\X$ (counted with multiplicity) by $\lambda_i$, $i=1,\ldots,n$, ordered by decreasing magnitude. 
It can be shown that the sign parameters $\{s_i\}$ of $\Sys$ obey the formula
\begin{align}
    \label{eq:hankeleigs}
    s_i=\textrm{sign}(\lambda_i), \qquad i=1,\ldots,n.
\end{align}
These facts can be obtained by combining results from~\cite{maciejowski1988balanced} with~\cite[Lemma~5.6]{Ant05}.
It can be shown that the Hankel singular values of a system $\Sys$ are the unsigned eigenvalues of its cross Gramian~\cite[Section~2.7]{maciejowski1988balanced}, so the given ordering {imposed upon $\{\lambda_i\}_{i=1}^n$} is consistent with that of the system's Hankel singular values.
Using this observation we can check if $\Sys$ obeys the hypotheses of Theorem~\ref{thm:main} by inspecting the eigenvalues of its cross Gramian.

Ultimately, we will show the following: Given any asymptotically stable and minimal system having a realization satisfying the generalized sign-symmetry condition~\eqref{eq:gen_symm}, there exists a permutation $\pi$ of $(1,2,\ldots,n)$ such that
$$s_{\pi_i}=\wh{s}_{i},\quad i=1,\ldots,n.$$
That is, the sign parameters of $\Sys$ {(and thus, the sign-symmetry of the canonical balanced form)} can be deduced from the sign-symmetry property of \emph{any} realization of $\Sys$.
The observed ordering of these sign parameters (derived from the decreasing magnitude of the Hankel singular values to which they correspond) varies. This variation follows from the fact that any realization of a system that satisfies~\eqref{eq:gen_symm} is \textit{not unique}. 
Indeed given one such realization, others can be obtained via symmetric permutations of the state space. This situation is akin to Remark~\ref{remark:bal_ordering}: if one were to observe a balanced realization of a system $\Sys$ satisfying the sign-symmetry condition~\eqref{eq:signA}, it would not be clear whether the Hankel singular values were ordered non-increasingly in that basis.
One can show that there does exist a \textit{permuted} realization in which the states are ordered corresponding to the canonical ordering of the system's sign parameters (i.e., where the permutation $\pi$ is the identity).

We are now prepared to state the main result of this section.

\begin{theorem}
    \label{thm:main2}
    Let $\Sys$ be an order-$n$ asymptotically stable and minimal SISO system, as in~\eqref{eq:system}. Suppose $\Sys$ has a (not necessarily balanced) realization $(A,@b,@c)$ satisfying the generalized sign symmetry condition~\eqref{eq:gen_symm} with signs $\wh{s}_i\in\{\pm 1\}$ for $i=1,\ldots,n$. Then there exists a permutation $\pi=(\pi_1,\pi_2,\ldots,\pi_n)$ of $(1,2,\ldots,n)$ such that
    $$s_{\pi_i}=\wh{s}_{i},\quad i=1,\ldots,n.$$
    That is, the signs $\{\wh{s}_i\}_{i=1}^n$ are a permutation of the sign parameters $\{s_i\}_{i=1}^n$ of $\Sys$.
\end{theorem}

\begin{proof}
    Consider the realization~\eqref{eq:gen_symm} of $\Sys$.\ \  Applying the sign symmetry property and multiplying~\eqref{eq:sylveq} on the left by $\wh{S}$ reveals
    \begin{align*}
       0 = \wh{S}(A\X + \X A + bc)&=\wh{S}A\X + \wh{S}\X(\wh{S}A^\tr\wh{S})+ (\wh{S}b)(\wh{S}b)^\tr\\
       &=(\wh{S}A\wh{S})(\wh{S}\X) + (\wh{S}\X)(\wh{S}A^\tr\wh{S})+ (\wh{S}b)(\wh{S}b)^\tr. 
    \end{align*}
    Note that $(\wh{S}b)(\wh{S}b)^\tr$ is symmetric positive semidefinite; since $\Sys$ is asymptotically stable and minimal, Lyapunov's theorem~\cite[sect.~6.2]{Ant05} implies that $\wh{\X} := \wh{S}\X$ is symmetric positive definite. Then, we claim that the sign pattern of $\wh{S}$ determines the signs of the eigenvalues of $\X$. To see this, note that
    $$\lambda(\X)=\lambda(\wh{S}\wh{\X})=\lambda(\wh{\X}^{1/2}\wh{S}\wh{\X}\wh{\X}^{-1/2})=\lambda(\wh{\X}^{1/2}\wh{S}\wh{\X}^{1/2}).$$
    Thus $\wh{S}$ is a congruence transformation of $\X$, and hence by Sylvester's law of inertia \cite[Prop.~6.15]{Ant05} $\wh{S}$ and $\X$ have the same number of positive and negative eigenvalues. Denote the eigenvalues of $\X$ by $\lambda_1, \ldots, \lambda_n$, ordered in decreasing magnitude.  Recalling that $s_i=\textrm{sign}(\lambda_i)$ for $i=1,\ldots,n$, 
    there must exist some permutation $\pi=(\pi_1,\pi_2,\ldots,\pi_n)$ of $(1,2,\ldots,n)$ such that $$s_{\pi_i}=\wh{s}_i,\quad i=1,\ldots,n.$$
    The permutation stems from the fact that Sylvester's law of inertia determines the signs of the eigenvalues of $\X$; it does not tell us anything about their magnitude.
\end{proof}

This result implies that the sign pattern of $\wh{S}$ reveals the sign parameters of $\Sys$.\ \  It does not however indicate their ordering with respect to the Hankel singular values.
There exists a permuted realization that reveals the canonical ordering of the sign parameters of $\Sys$ (that is, the Hankel singular values to which they correspond are nonincreasing).
Next, we study a special class of systems having an \emph{arrowhead structure} for which (in certain cases) the result of Theorem~\ref{thm:main2} can be strengthened.

\section{Systems with arrowhead realizations}
\label{sec:specialcase}
Theorem~\ref{thm:main2} shows that the sign parameters of a linear system are determined by the generalized state-space symmetry property~\eqref{eq:gen_symm}. In this section, we apply the result of Theorem~\ref{thm:main2} to systems with \emph{arrowhead realizations}.
For such systems, the sign parameters can be derived directly from the off-diagonal entries of the corresponding state matrix.
At the end of the section, we apply this result to a {particular} class of systems having an arrowhead realization in which the {diagonal entries are negative and} the signs of the products of the off-diagonal entries are identical.
In this case, the sign structure of the arrowhead realization reveals the \emph{canonical ordering} of the system's sign parameters.
We begin with an example of such a system that arises as a model of the aggregate dynamics of a power network.
(Indeed, we first discovered the phenomena characterized in Theorems~\ref{thm:main} and~\ref{thm:main2} while studying this model.)

\subsection{A motivating example from power systems modeling}
\label{ss:ps_model}
A common technique for modeling the frequency dynamics of a network of coherent generators is to aggregate the system response into a single effective machine. It is shown in~\cite{min2019accurate,min2020accurate} that for a network of $M$ coherent generators, the aggregate frequency dynamics are approximated well by an order $n=M+1$ linear system $\Sys$ having the transfer function
\begin{equation}\label{eq:powertf}
G(s) = \frac{1}{\wh{m}s + \wh{d} + \sum_{i=1}^M \frac{r_i^{-1}}{\tau_i s + 1}},
\end{equation}
for generators given by the swing model with first-order turbine control. Here $\wh{m}$ and $\wh{d}$ denote the aggregate inertia and damping coefficients of the generators in the network, while $\tau_i$ and $r_i^{-1}$ denote the time constant and droop coefficient of the $i$th generator. For the theoretical justification that $G(s)$ sufficiently approximates the network response, see~\cite[Sec.~2]{min2019accurate}.

There is a natural coordinate system for the model~\eqref{eq:powertf} in which the dynamics of $\Sys$ are described by simple expressions involving the physical parameters of the network. (This is obtained by applying~\cite[Equation~1]{min2019accurate} to the dynamics~\cite[p.~3009, Example~2]{paganini2019global} and simplifying.)
In these coordinates, the realization of $\Sys$ has an \emph{arrowhead form}:
\begin{align}
\label{eq:arrowhead_ps}
A=\begin{bmatrix}
 -\wh{d}& \sqrt{\tfrac{r_1^{-1}}{\wh{m}\tau_1}} &\dots &\sqrt{\tfrac{r_M^{-1}}{\wh{m}\tau_M}}\\[.5mm]
 -\sqrt{\tfrac{r_1^{-1}}{\wh{m}\tau_1}} & -\tfrac{1}{\tau_1} & &  \\
 \vdots & & \ddots &  \\
-\sqrt{\tfrac{r_M^{-1}}{\wh{m}\tau_M}}&  &   & -\tfrac{1}{\tau_M} 
\end{bmatrix},\ \ \ b = \frac{1}{\sqrt{\wh{m}}}e_1, \ \ \ c=\frac{1}{\sqrt{\wh{m}}}e_1^{\tr}.
\end{align}
Apart from the first row, first column, and main diagonal,
all entries of $A$ are zero.
Such an $A$ is called an \textit{arrowhead matrix}.
In this application, the parameters $\wh m$, $\wh d$, $\{r_i^{-1}\}$, and $\{\tau_i\}$ are guaranteed to be positive. Evidently, the sign pattern of~\eqref{eq:arrowhead_ps} obeys the symmetry condition~\eqref{eq:gen_symm} for $\wh{S}=\mbox{diag}(+1,-1,\ldots,-1)\in\R^{n\times n}$. (Note that this condition is entirely determined by the signs of the off-diagonal entries of $A$.)\ 
For these systems we first observed numerically that the balanced truncation error bound was consistently tight for all orders of reduction, and then noticed that the truncated system in $\Sys$'s canonical balanced form obeyed the sign pattern described in~\eqref{eq:s2}. That is, we observed numerically that $s_1=+1$ and $s_i=-1$ for all $i=2,\ldots,n=M+1$, suggesting that the permutation $\pi$ given in the statement of Theorem~\ref{thm:main2} is always the identity. 
In other words, the sign symmetry of~\eqref{eq:arrowhead_ps} reveals the sign parameters of $\Sys$ \emph{in the canonical order} according to the Hankel singular values.
Motivated by this example, we apply the result of Theorem~\ref{thm:main2} to a more general class of systems with arrowhead realizations.

\subsection{Arrowhead systems and preliminary tools}

In general, we say a SISO LTI system $\Sys$ has an \textit{arrowhead realization} if it has a realization of the form
\begin{align}
\label{eq:arrowheadrep}
   A=\begin{bmatrix}
            d_1 & {\alpha_2} &\cdots &  {\alpha_n}\\
            {\beta_2} & d_2  \\
            \vdots & & \ddots & \\
            {\beta_n} &  &   &  d_n
            \end{bmatrix},\ \ \ b = \gamma e_1, \ \ \ \textrm{and} \ \ \ c=e_1^{\tr},
\end{align}
where $\gamma\in \R$ is nonzero, and $e_i$ denotes the $i$th column of the identity matrix.
Occasionally, we refer to systems having arrowhead realizations as \textit{arrowhead systems}.
Recalling Figure~\ref{figure:arrowheadnetwork}, we emphasize that these arrowhead realizations describe a natural physical structure if we consider $A$ to be the weighted adjacency matrix of a directed graph.

Before strengthening the result of Theorem~\ref{thm:main2} for a class of arrowhead systems (including the power systems model that motivated this investigation) we present some useful preliminary facts and results concerning arrowhead systems.
We first provide a formula for the inverse of an arrowhead matrix, as it will serve as a useful tool throughout this section.
Let $\alpha=[\alpha_2\  \cdots \ \alpha_n]\in\R^{1\times (n-1)}$, 
$\beta=[\beta_2\ \cdots \ \beta_n]^{\tr}\in\R^{(n-1)\times 1}$, and $D=\textrm{diag}(d_2,\,\dots,\,d_n)\in \R^{(n-1)\times(n-1)}$.
If $A\in \R^{n\times n}$ has the arrowhead form~\eqref{eq:arrowheadrep} with $d_i\neq 0$ for all $i=2,\dots,n$, and $d_1-\alpha D^{-1}\beta \neq 0$, then the inverse of $A$ can be expressed as a diagonal matrix plus a rank-one update~\cite{salkuyeh2018explicit}:
\begin{align}
    \label{eq:aheadinv}
    A^{-1}=\begin{bmatrix} \phantom{^1}0 & 0 \\ \phantom{^1}0 & D^{-1} \\  \end{bmatrix} + \rho \begin{bmatrix} -1 \\   D^{-1}\beta \end{bmatrix}\begin{bmatrix} -1 & \alpha D^{-1} \end{bmatrix},
\end{align}
where 
$$\rho = \frac{1}{d_1 - \alpha D^{-1}\beta}.$$
(This expression follows from the Sherman--Morrison--Woodbury formula~\cite[p.~65]{GV12}; though most easily derived when $d_1\ne 0$, the formula holds even when $d_1=0$.) 
While we have stated the formula for arrowhead matrices having real entries,~\eqref{eq:aheadinv} holds when the nonzero entries of $A$ are complex-valued, as well.

Applying~\eqref{eq:aheadinv} to the transfer function $G(s)=c(sI-A)^{-1}b$ of a system having an arrowhead realization~\eqref{eq:arrowheadrep} leads to the revealing form
\begin{align}
    \label{eq:aheadtf}
    G(s) = \frac{\gamma}{s - d_1 - \sum_{i=2}^n\frac{\alpha_i\beta_i}{s-d_i}}.
\end{align}
Any system having a transfer function $G(s)$ with the general form~\eqref{eq:aheadtf} has an arrowhead realization given by~\eqref{eq:arrowheadrep}.
Conversely, given any system having an arrowhead realization~\eqref{eq:arrowheadrep}, the transfer function $G(s)$ can be expressed in the form~\eqref{eq:aheadtf}. 
At the end of this section, we use this representation to derive an arrowhead form for linear systems under mild assumptions.

Next, we prove conditions under which a given arrowhead realization~\eqref{eq:arrowheadrep} of a system is \emph{minimal} {(that is, both reachable and observable)} based on the entries of the corresponding arrowhead matrix. {This will be necessary to satisfy the hypotheses of Theorem~\ref{thm:main2}.
}
\begin{lemma}
\label{lemma:reach_obsv_arrowhead}
Let $\Sys$ be an order-$n$ asymptotically stable SISO system as in~\eqref{eq:system} with an arrowhead realization~\eqref{eq:arrowheadrep}. {The following equivalences  hold:}
\begin{enumerate}
    \item {$(A,c)$ is observable, if and only if $d_i\neq d_j$
for $i\neq j$ and $\alpha_j\neq 0$ for all $j=2,\ldots,n$.}
    \item {$(A,b)$ is reachable, if and only if $d_i\neq d_j$ for $i\neq j$ and $\beta_j\neq 0$ for all $j=2,\ldots,n$.}
\end{enumerate}
\end{lemma}
\begin{proof}
By duality, $(A,b)$ is reachable if and only if $(A^\tr,b^\tr)$ is
observable. Thus, it suffices to prove the observability
statement.
{By the Hautus test~\cite[Theorem 4.15]{Zho96}, the pair $(A,c)$ is not
observable if and only if there exists an eigenvector $v\neq 0$ of $A$
such that $c^\tr v=0$, i.e., $v_1=0$.}

We first prove the necessity of the conditions {by the constrapositive implication}.
If $\alpha_j=0$, for some $j\in\{2,\ldots,n\}$, then $v=e_j$ is an
eigenvector of $A$ with $v_1=0$. 
If $d_i=d_j$ and $\alpha_i\neq 0$, $\alpha_j\neq 0$ for $i\neq j$,
then $\alpha_je_i-\alpha_i e_j$  is an
eigenvector of $A$ with $v_1=0$. 
In both cases $(A,c)$ is not observable by the previous application of the Hautus test.

We now {prove} sufficiency of the conditions {by contradiction}. Suppose  $d_i\neq d_j$ for $i\neq j$ and $\alpha_j\neq 0$ for $j=2,\ldots,n$, {and additionally suppose that $(A,c)$ is not observable. Thus there exists an eigenvector $v\neq 0$ of $A$ such that $c^\tr v = 0$, that is $v_1=0$.}
The eigenvector condition $Av=\lambda v$ with $v_1=0$ implies $d_jv_j=\lambda v_j$ for
$j=2,\ldots,n$. Hence, $v_i\neq 0$ for exactly one
$i\in\{2,\ldots,n\}$. Without loss of generality $v_i=1$, i.e., $v=e_i$.
But then $Av={Ae_i}=\alpha_ie_1+d_ie_i$, {and so $v$} is not
an eigenvector of $A$ as $\alpha_i\neq 0$. It follows that $(A,c)$ is observable.
\end{proof}

From Lemma~\ref{lemma:reach_obsv_arrowhead}, it follows straightforwardly that an arrowhead system is minimal if and only if $\alpha_i$ and $\beta_i$ are nonzero for all $i=2,\ldots,n$ and $d_i\neq d_j$ for all $i\neq j$.
We are now in a position to apply Theorem~\ref{thm:main2} to these arrowhead systems.

\begin{corollary}
\label{cor:arrowhead_case}
Let $\Sys$ be an order-$n$ asymptotically stable and minimal SISO system as in~\eqref{eq:system} with an arrowhead realization~\eqref{eq:arrowheadrep}.
Then there exists a permutation $\pi=(\pi_1,\pi_2,\ldots,\pi_n)$ of $(1,2,\ldots,n)$ such that the sign parameters of $\Sys$ are given by
\begin{align}
\label{eq:signformula}
s_{\pi_1} = \textrm{sign}(\gamma), \quad\mbox{and}\quad s_{\pi_i} = \textrm{sign}(\gamma\alpha_i\beta_i), \quad i=2,\ldots,n.
\end{align}
\end{corollary}

\begin{proof}
    Note that any such realization~\eqref{eq:arrowheadrep} can be made to satisfy the sign symmetry condition~\eqref{eq:gen_symm} via an appropriate diagonal change of basis. In particular, let
\begin{equation}
D=\sqrt{|\gamma|}\ \mbox{diag}\left(1,\sqrt{\tfrac{|\beta_2|}{|\alpha_2|}},\ldots,
   \sqrt{\tfrac{|\beta_n|}{|\alpha_n|}}\right).
   \label{eq:D}
\end{equation}
Then it follows that the transformed system
\begin{align*}
   A_s&=D^{-1}AD=\left[
    \begin{array}{cccc}
      d_1&\mbox{sign}(\alpha_2)\sqrt{|\alpha_2\beta_2|}&\ \cdots\ &\mbox{sign}(\alpha_n)\sqrt{|\alpha_n\beta_n|}\\[6pt]
      \mbox{sign}(\beta_2) \sqrt{|\alpha_2\beta_2|}&d_2\\[3pt]
      \vdots&&\ddots\\[3pt]
      \mbox{sign}(\beta_n) \sqrt{|\alpha_n\beta_n|}&&&d_n
    \end{array}
  \right],\\[6pt]  
b_s&=D^{-1}b=\mbox{sign}(\gamma)\sqrt{|\gamma|}e_1, \quad c_s=cD=\sqrt{|\gamma|}e_1^\tr
 \end{align*}
satisfies~\eqref{eq:gen_symm}
with 
$$\wh{S} = \textrm{diag}\left(\textrm{sign}(\gamma), \  \textrm{sign}(\gamma@\alpha_2\beta_2), \ \ldots,\  \textrm{sign}(\gamma@\alpha_n\beta_n)\right),$$
thus we have $A_s = \wh{S}A_s^\tr\wh{S}$ and $b_s = (c_s \wh{S})^\tr$.
In other words, the entries of the sign matrix $\wh{S}$ above are entirely determined by the signs of $\gamma$ and the off-diagonal entries $\{\alpha_i\}$ and $\{\beta_i\}$ of the arrowhead matrix.
Since $\Sys$ is asymptotically stable and minimal by hypothesis, we can apply Theorem~\ref{thm:main2} to obtain the stated result.
\end{proof}
{Corollary~\ref{cor:arrowhead_case} is stated for systems having the particular structure in~\eqref{eq:arrowheadrep}.  While this structure differs slightly from the power system model in~\eqref{eq:arrowhead_ps}, the two models are related by the simple diagonal state space transformation~\eqref{eq:D}, and the corollary applies to both.}

\subsection{A special case of arrowhead systems}
Recall the power systems model presented in Subsection~\ref{ss:ps_model}.
For these systems we first observed numerically that the balanced truncation error bound was consistently tight, and then noticed that the truncated system in $\Sys$'s canonical balanced form obeyed the sign pattern described in~\eqref{eq:s2}.
Applying Corollary~\ref{cor:arrowhead_case} to this power system model $\Sys$ with realization~\eqref{eq:arrowhead_ps} guarantees that
\[s_{\pi_1} = \mbox{sign}\left(\tfrac{1}{\sqrt{\wh{m}}}\right)=+1,\quad s_{\pi_{i}} = \mbox{sign}\left( -\sqrt{\tfrac{r_{i-1}^{-1}}{\wh m\tau_{i-1}}} \right)=-1, \quad i =2,\ldots, M+1.\]
However, for these systems, we observed numerically that $s_1=+1$ and $s_i=-1$ for all $i=2,\ldots,n=M+1$, suggesting that the permutation $\pi$ given in the statement of Corollary~\ref{cor:arrowhead_case} is the identity.
We next show that this pattern holds in a more general setting.
Specifically, if a system has an arrowhead realization~\eqref{eq:arrowheadrep} such that {the diagonal entries as well as} the signs of the products of the off-diagonal entries are \emph{negative}, i.e., {$d_i <0$} and $\textrm{sign}(\alpha_i\beta_i)=-1$ for all $i=2,\ldots,n$, then the trailing $n-1$ sign parameters of $\Sys$ are identical.

\begin{corollary}
\label{cor:onesignflip} {Let $\Sys$ be an order-$n$ asymptotically stable and minimal SISO system as in~\eqref{eq:system} with an arrowhead realization as in~\eqref{eq:arrowheadrep}. Suppose further that the arrowhead realization~\eqref{eq:arrowheadrep} is such that $\textrm{sign}(\alpha_i\beta_i)=-1$ for all $i=2,\ldots,n$ and {$d_i < 0$ for all $i = 1,\ldots,n$.}
Then the sign parameters of $\Sys$ are given by
$$s_1 = \textrm{sign}(\gamma),\quad s_i = -\textrm{sign}(\gamma), \quad i=2,\ldots,n.$$
{Additionally, 
$\Sys$ has distinct Hankel singular values.}}
\end{corollary}

\begin{proof}
Without loss of generality suppose that $\mbox{sign}(\gamma)=+1$.
If the given arrowhead realization of $\Sys$ is such that the permutation $\pi$ in Corollary~\ref{cor:arrowhead_case} is the identity, then the result follows trivially.
Otherwise, Corollary~\ref{cor:arrowhead_case} only guarantees that $\Sys$ has one positive sign parameter and $n-1$ negative sign parameters.
We claim that the positive sign parameter corresponds to the dominant eigenvalue of the cross Gramian, i.e., $s_1 = +1$. 
As noted in~\cite[Remark~5.4.3]{Ant05},  $2\,\trace(\mathcal{X})=-\trace(cA^{-1}b)$. 
Since $b=\gamma e_1$ and $c^{\tr}=e_1$, formula~\eqref{eq:aheadinv} gives 
$$cA^{-1}b=\frac{\gamma}{d_1-\sum_{i=2}^{n}\frac{\alpha_i\beta_i}{d_i}}.$$
Since this quantity is a scalar, $cA^{-1}b=\trace(cA^{-1}b)$.
{The hypothesis that $d_i<0$ for all $i=1,\ldots,n$} along with the assumption that $\mbox{sign}(\gamma)=+1$ implies
$$\trace(cA^{-1}b)=cA^{-1}b = \frac{\gamma}{d_1 - \sum_{i=2}^{n}\frac{\alpha_i\beta_i}{d_i}}<0.$$
Thus $\trace(\mathcal{X})=-\frac{1}{2}\trace(cA^{-1}b)>0$. 
Because the trace of $\mathcal{X}$ is equal to the sum of its eigenvalues and {$\mathcal{X}$ has only one positive eigenvalue}, the dominant eigenvalue of $\mathcal{X}$ must be the positive one. Thus we conclude that $s_1=+1$ and $s_i=-1,$ for all $i=2,\dots,n$. 
{Theorem~\ref{thm:main} further implies that the $n-1$ trailing Hankel singular values of $\Sys$, and thus all Hankel singular values of $\Sys$, are distinct.}
For the case of $\mbox{sign}(\gamma)=-1$, the proof follows nearly identically
by noting that $\trace(\X) = -\frac{1}{2}\trace(cA^{-1}b)<0$, which shows that $s_1 = -1$ and $s_i=+1$ for all $i=2,\ldots,n$ by the same logic as above.
\end{proof}

Corollary~\ref{cor:onesignflip} implies that the trailing sign parameters of the systems described here are consistent, and can be determined directly from the off-diagonal entries of the arrowhead matrix using the formula
$$s_1=\mbox{sign}(\gamma),\quad s_i = \mbox{sign}(\gamma\alpha_i\beta_i),\quad i=2,\ldots,n.$$
In other words, the sign parameters of systems having an arrowhead realization satisfying the hypotheses of Corollary~\ref{cor:onesignflip} exhibit a single sign flip after the first parameter, and so these systems obey the necessary sign consistency conditions in Theorem~\ref{thm:main} for \textit{all} orders of reduction: the balanced truncation $\mathcal{H}_\infty$ error bound~\eqref{eq:bound} \emph{will always be hold with equality.}

\begin{example}
\label{ex:ps}
We illustrate {Theorem~\ref{thm:main} and Corollary~\ref{cor:onesignflip}} with a particular example of the power system model presented in Subsection~\ref{ss:ps_model}.
Consider a network with $M=4$ coherent generators, giving a SISO LTI system of order $n=M+1=5$. Take $\wh{m}=0.044$, $\wh{d}=0.038$, 
$$(r_1^{-1}, r_2^{-1}, r_3^{-1}, r_4^{-1})  
= (0.013,\ 0.014,\ 0.022,\ 0.025),$$
and 
$$
(\tau_1, \tau_2, \tau_3, \tau_4) = 
(5.01,\ 6.82,\ 7.38,\ 7.79).
$$
Since the physical parameters associated with the system are all positive, Corollary~\ref{cor:onesignflip} implies, for the realization of $\Sys$ given by~\eqref{eq:arrowhead_ps}, that the sign parameters are
\[s_{1} = \mbox{sign}\left(\tfrac{1}{\sqrt{\wh{m}}}\right)=+1,\quad s_{i} = \mbox{sign}\left( -\sqrt{\tfrac{r_{i-1}^{-1}}{\wh m\tau_{i-1}}} \right)=-1, \quad i =2,\ldots, 5.\]
These signs can be verified by computing a balanced realization of $\Sys$ 
satisfying the symmetry condition~\eqref{eq:signA} with the sign matrix
\[S=\textrm{diag}\,(1,-1,-1,-1,-1).\]
The Hankel singular values of $\Sys$ with transfer function $G(s)$ in~\eqref{eq:powertf} are
\[\Sigma = \textrm{diag}(11.63, 7.13, 3.53\times10^{-2}, 8.48\times 10^{-5}, 4.12\times 10^{-8}).\]
Now we construct the canonical balanced realization of $\Sys$ using the formula~\eqref{eq:canonicalform}, then compute order-$r$ balanced truncation approximations to $\Sys$ for $r=2$, $3$, and $4$. Under these conditions the truncated systems obey the sign consistency in~\eqref{eq:s2}. 
As in Example~\ref{ex:thm1}, we highlight this symmetry by partitioning the system for $r=3$:
\begin{align*}
A&= \left[\begin{array}{rrr|rr}
-0.9913  &  0.5924 &  -0.0467  &  0.0020  &  0.0000\\
   -0.5924 & -0.0216  &  0.0087 & -0.0004 &  -0.0000\\
    0.0467 &   0.0087 &  -0.1800  &  0.0157&  0.0003\\\hline
   -0.0020 &  -0.0004  &  0.0157  & \hot{-0.1437} & \hot{-0.0062}\\
   -0.0000 &  -0.0000  &  0.0003  & \hot{-0.0062} &  \hot{-0.1372}
\end{array}\right],\\
b &= \left[
\begin{array}{r}-4.8009\\
   -0.5552\\
    0.1126\\\hline
   \hot{-0.0049}\\
   \hot{-0.0001} \end{array}\right],\ \ \
   c^{\tr} =  \left[\begin{array}{r}-4.8021 \\
   0.5552 \\
    -0.1126 \\\hline
   \hot{0.0049} \\
   \hot{0.0001} \end{array}\right],\ \ \ d=0.
\end{align*} 
Table~\ref{tab:powersystem_errors} compares the $\mathcal{H}_\infty$ norm of the error system to the balanced truncation upper bound~\eqref{eq:bound}.  Because the trailing sign parameters of $\Sys$ are all $-1$, we can perform truncation at any order $r\geq 1$ and the truncated system will satisfy the sign requirements~\eqref{eq:s2} of Theorem~\ref{thm:main}. Thus the balanced truncation error bound  holds with equality for approximations of all orders.

\begin{table}[hh]
    \caption{\label{tab:powersystem_errors}
     $\mathcal{H}_\infty$ norm of the error system, compared to the balanced truncation upper bound~\eqref{eq:bound} for a power system, {$M=4$}.}
    \centering
    \vspace{4mm}
    \begin{tabular}{c|cc}
    & $\|{G}-{G}_r\|_{\mathcal{H}_\infty}$ & $2(\sigma_{r+1}+\dots+\sigma_n)$\\[2pt] \hline 
    \rule[-4pt]{0pt}{12pt}
    $r=1$ & $1.747\times 10^{1\phantom{-}}$ & $1.747\times 10^{1\phantom{-}}$\\
   \rule[-4pt]{0pt}{12pt}
         $r=2$ & $7.067 \times 10^{-2}$ & $7.067 \times 10^{-2}$\\
     \rule[-4pt]{0pt}{12pt}
           $r=3$ & $1.697 \times 10^{-4}$ & $1.697 \times 10^{-4}$\\
      \rule[-4pt]{0pt}{12pt} 
           $r=4$ &  $8.248\times 10^{-8}$ &  $8.248\times 10^{-8}$\\
        \hline
    \end{tabular}
\end{table}
\end{example}

To motivate our study of arrowhead systems in the general setting, we derive an arrowhead form for linear systems having distinct zeros.


\begin{proposition}
\label{prop:aheadform}
Let $\Sys$ be an order-$n$ asymptotically stable and minimal SISO system as in~\eqref{eq:system} such that the associated transfer function $G(s)$ has $n-1$ distinct zeros, $d_2,\ldots,d_n\in\C$.\ \  Then $\Sys$ has an arrowhead realization of the form
$$A= \begin{bmatrix}
    -{\gamma}d_1 & {\gamma}\alpha\\
    -\mathbf{1} & D
    \end{bmatrix},\quad b = {\gamma}e_1 \quad \mbox{and} \quad  c=e_1^{\tr},$$
where $D=\textrm{diag}(d_2,\,\ldots,d_n)\in \mathbb{C}^{(n-1)\times(n-1)}$, $\mathbf{1}\in\R^{n-1}$ is the vector containing all ones, $\alpha=[\alpha_2\,\cdots \alpha_n]\in \C^{1\times(n-1)}$ with $\alpha_i\neq 0$ for all $i=2,\ldots,n$, and $\gamma\in \C$ is nonzero.
\end{proposition}

\begin{proof}
This derivation of the arrowhead form follows transfer function manipulations for power systems from~\cite[Section~III.C]{min2020accurate}, but applied here to a broader class of systems.
Let $G(s)=N(s)/D(s)$.  By minimality of $\Sys$ and the assumption that $G(s)$ has $n-1$ distinct zeros, $N(s)$ and $D(s)$ are polynomials of degree $n-1$ and $n$, respectively, and have no common factors.
By the polynomial division algorithm,
$D(s)=N(s)Q(s) + R(s)$, where $Q$ is a degree-1 polynomial and  $0<\mbox{deg}(R)<n$. (Note that $\mbox{deg}(R)=0$ would contradict the minimality of $\Sys$, as it would imply that $N(s)$ and $D(s)$ have a common factor.)
Write $Q(s) = \mu s+d_1$ for some (possibly zero) $d_1\in \mathbb{C}$ {and nonzero $\mu\in\C$}. 
Since the zeros $d_2,\ldots, d_n$ of $G(s)$, and hence of $N(s)$, are distinct, the partial fraction expansion of the rational function $R(s)/N(s)$ takes the form
\[ R(s)/N(s) = \sum_{i=2}^n \frac{\alpha_i}{s-d_i}.\]
We can thus write
\begin{align*}
    G(s)=\frac{N(s)}{D(s)}&=\frac{N(s)}{N(s)Q(s)+R(s)}\\
    &=\frac{1}{{\mu}s+d_1 + R(s)/N(s)}\\
    &=\frac{1}{{\mu}s+d_1 + \sum_{i=2}^{n}\frac{\alpha_i}{s-d_{i}}}\\
    &=\frac{{\gamma}}{s+{\gamma}d_1 +{\gamma}\sum_{i=2}^{n}\frac{\alpha_i}{s-d_{i}}},
\end{align*} 
{by dividing out $\mu$ and defining $\gamma = 1/\mu$.} 
Define
$$A= \begin{bmatrix}
    -{\gamma}d_1 & {\gamma}\alpha\\
    -\mathbf{1} & D
    \end{bmatrix},\quad b = e_1 \quad \mbox{and} \quad  c=e_1^\tr,$$
where $D=\textrm{diag}(d_2,\dots,d_n)\in \mathbb{C}^{(n-1)\times(n-1)}$ and $\alpha=[\alpha_2\,\cdots\,\alpha_n]\in\mathbb{C}^{1\times(n-1)}$. Using the formula~\eqref{eq:aheadinv} to get the $(1,1)$ entry of the inverse of an arrowhead matrix, one can show that $c@(sI-A)^{-1}b = N(s)/D(s) = G(s)$, proving that $(A,b,c)$ is a valid realization of $\mathcal{G}$.\ \ 
The fact that $\alpha_i\neq 0$ for all $i=2,\ldots,n$ follows necessarily from Lemma~\ref{lemma:reach_obsv_arrowhead} and the minimality assumption.
\end{proof}


\section{Conclusion}\label{sec:conc}
We have shown that the balanced truncation error bound~\eqref{eq:bound} holds with equality for SISO systems satisfying the sign consistency condition~\eqref{eq:s2}, providing an explicit formula for the approximation error in terms of the system's Hankel singular values.
This analysis generalizes an earlier result for state-space symmetric systems.
We additionally proved that the sign parameters corresponding to a system's Hankel singular values are determined by the generalized state-space symmetry property~\eqref{eq:gen_symm} of the system.
We showed how this result can be strengthened for a special class of arrowhead systems, illustrated by a model of coherent generators in power systems.
From these results, one can verify whether the sign consistency condition~\eqref{eq:s2} in Theorem~\ref{thm:main} holds, and thus whether or not the corresponding order-$r$ balanced truncation approximation achieves the $\mathcal{H}_\infty$ error bound~\eqref{eq:bound}.






\section*{Acknowledgments}

The authors thank Christopher Beattie and Vassilis Kekatos for helpful discussions, Christian Himpe bringing reference~\cite{6949058} to our attention, and the referees for numerous comments that improved our presentation.

\section*{Declarations}

\subsection*{Funding}
This work was supported by US National Science Foundation grant AMPS-1923221.

\subsection*{Conflict of interest/Competing interests}

The authors declare that they have no conflict of interest.

\addcontentsline{toc}{section}{References}
\bibliographystyle{plainurl}
\bibliography{references}


\end{document}